\documentclass[lettersize,journal]{IEEEtran}
\usepackage{algorithmic}
\usepackage{amsmath,amssymb,amsfonts,amsthm, amssymb}
\usepackage{mathtools}
\usepackage{array}
\usepackage{stfloats}
\usepackage{graphicx}
\usepackage{url}
\usepackage{hyperref}
\usepackage{verbatim}
\usepackage{cite}
\usepackage{textcomp}
\usepackage{xcolor}
\usepackage{tabularx}
\usepackage[linesnumbered,ruled]{algorithm2e}
\newtheorem{theorem}{Theorem}
\newtheorem{proposition}{Proposition}

\newcommand{\Sum}[2]{\sum_{#1}^{#2}}

\begin{document}
\title{Decentralized Task Offloading and Load-Balancing for Mobile Edge Computing in Dense Networks}
\author{{Mariam Yahya\IEEEauthorrefmark{1}, Alexander Conzelmann\IEEEauthorrefmark{1}, and Setareh Maghsudi}
\thanks{\IEEEauthorrefmark{1} Equal contribution}

\thanks{This research was supported by Grant MA 7111/6-1 from the German Research Foundation (DFG) and Grant 16KISK035 from the German Federal Ministry of Education and Research (BMBF).} 

\thanks{M. Yahya and A. Conzelmann are with the Department of Computer Science, University of Tübingen, 72074 Tübingen, Germany (email: mariam.yahya@uni-tuebingen.de; a.conzelmann@student.uni-tuebingen.de).} \thanks{S. Maghsudi is with the Faculty of Electrical Engineering and Information Technology, Ruhr-University Bochum, 44801 Bochum, Germany (email: setareh.maghsudi@ruhr-uni-bochum.de).}      } 
\maketitle
\begin{abstract}
We study the problem of decentralized task offloading and load-balancing in a dense network with numerous devices and a set of edge servers. Solving this problem optimally is complicated due to the unknown network information and random task sizes. The shared network resources also influence the users' decisions and resource distribution. Our solution combines the mean field multi-agent multi-armed bandit (MAB) game with a load-balancing technique that adjusts the servers' rewards to achieve a target population profile despite the distributed user decision-making. Numerical results demonstrate the efficacy of our approach and the convergence to the target load distribution.
\end{abstract}

\begin{IEEEkeywords}
Edge computing, mean field game, multi-armed bandits, load-balancing.
\end{IEEEkeywords}
\section{Introduction}
\IEEEPARstart{T}{HE} stringent requirements of the emerging 5G/6G networks that necessitate low latency and high computational capabilities drove the evolution of mobile edge computing (MEC) technology. In MEC, computational and storage resources are strategically positioned near end users with limited resources in a communication network.
Computation task offloading is an essential element of the MEC paradigm in which resource-limited devices offload their computation-intensive tasks for processing at edge servers with low computation and communication delays. Typically, MEC networks consist of a substantial number of users and servers, thereby raising the critical issue of effectively allocating communication and computation resources to optimize network performance. Another challenge is balancing the load among the network's edge servers to avoid prolonged service delays, which reduce the quality-of-service and decrease the system throughput \cite{JAFARNEJADGHOMI201750}. In some variations, computation costs or authorizations at different servers may vary, imposing specific user loads on the servers. In centralized systems, a controller oversees incoming tasks and assigns them to servers, as exemplified by the Hadoop MapReduce software \cite{dsouza2015survey}. Alternatively, in decentralized load-balancing, agents autonomously allocate their tasks to the servers \cite{liu2021game}. These algorithms can be either dynamic or static, depending on the current environment.

Optimal task offloading and load-balancing are further compounded in dynamic environments and in the absence of a centralized decision-maker. Therefore, solving the task offloading problem given precise network information, as in \cite{shi2023} among many others, is unrealistic in dynamic and uncertain environments. Consequently, some authors use online learning methods such as multi-armed bandits (MABs) to enable learning the network parameters and deciding about task offloading accordingly. References \cite{Cheng23:DTM} and \cite{wang2022decentralized} apply multi-agent MABs for users with homogeneous and heterogeneous offloading requirements, where the reward considers the intrinsic task value and delay cost. Additionally, a budgeted MAB model can jointly minimize the delay and the energy cost \cite{Ghoorchian2021}.

Compared to the state-of-the-art, our work offers several novelties. For example, \cite{Cheng23:DTM,wang2022decentralized,Ghoorchian2021} assume equal resource availability for all devices and disregard potential competition over shared resources. Additionally, such approaches do not apply to dense networks due to excessive complexity. Furthermore, they can result in an undesired distribution of devices across the servers, thus degrading the network performance.

The mean field game (MFG) theory studies differential games with a large population of decision makers \cite{banez2021mean}. Although its seminal model does not include uncertainty, MFG can be combined with learning to accommodate the lack of information. Reference \cite{gummadi2013mean} proposes an approach for analyzing an MAB game with numerous agents using a mean field framework. It shows that a non-stationary environment appears stationary by approximating the agents' interactions using their long-term average. This mean field MAB model was used in energy harvesting ultra-dense small cell networks to address the problem of decentralized user association based on their harvested energy amount and the population profile on the shared channels \cite{maghsudi2017distributed}. In this paper, we use the mean field MAB model in conjunction with load-balancing techniques for task allocation, mainly based on network delays.

We study dynamic decentralized task offloading and load-balancing in a dense network comprising several devices and a set of edge servers. 
The primary objective of each device is to receive the result of the task before an expiry time \cite{mancuso2023efficiency}. This problem is complex as the delay is prone to high uncertainty due to the randomness in the communication channel, the task size, and the server's load. Additionally, since the devices share the wireless bandwidth and computation resources, the amount of resources allocated to each device depends on the actions of the others. To address this, we propose a novel approach: modeling the problem using multi-agent MABs, where each device is an agent aiming to identify the best server. The reward obtained by an agent depends on the arm and the actions of other agents. Under such a scheme, it is unrealistic to analyze the performance of each agent individually, and so the agents' interactions are viewed as a dynamic game \cite{gummadi2013mean}. 
 To enhance this game, we propose an agent-based dynamic load-balancing algorithm that guides the agents to achieve a target distribution over the servers in the mean field steady state. This target distribution is determined centrally by the network manager based on the servers' heterogeneous costs or authorizations as the servers often belong to different operators. However, the task offloading decisions are decentralized \cite{mancuso2023efficiency} and thus scalable in dense networks thanks to less complex interactions between the devices and a lower cost of acquiring information than centralized methods.
\section{System Model}
\label{sec:system_model}
We consider a dense network with $n$ edge servers and $m$ devices that offload heavy computational tasks of random sizes to the servers. The size of the task offloaded by device $k{\in} \{1, \dots, m\}$, denoted by $s_k$, follows a truncated normal distribution in the range $[s_a, s_b]$ \cite{chen2019mobile}. The devices compute tasks smaller than $s_a$ locally, and $s_b$ is the maximum task size that they can offload at once. Besides, to manage the dynamic nature of the problem, the devices use a time-slotted system to offload their computational tasks. At the beginning of a time slot~$t$, each device~$k$ initiates the task offloading to an edge server. The channel between a device and server $ i {\in} \{1, \dots, n \}$ is  Rayleigh fading with bandwidth $B_{ki,t} {=} B/(m f_{i,t})$, where $B$ is the bandwidth of the uplink channel, equal for all servers. $f_{i, t}$ is the fraction of devices that offload to server~$i$ at time $t$. The vector $\boldsymbol{f}_{i,t} {=}[f_{1,t}, \dots, f_{n,t}]$ forms the population profile, which is the distribution of the devices on the servers. 
The size of the outcome of the task offloading process is proportional to the size of the offloaded task $s_k$, with proportionality constant $\rho$ \cite{shi2023}. The fading on the downlink and uplink channels is equal due to channel reciprocity, and the bandwidth of the downlink channel is constant. The edge servers have the same processing speed, each operating at a rate of $F$ cycles/second. Each server can support running multiple tasks in parallel using virtualization technologies \cite{guo2022joint}.

The delay experienced by a device~$k$ offloading to edge server~$i$ comprises three primary components: the uplink delay $d^{\text{UL}}_{ki,t}$, the downlink delay $d^{\text{DL}}_{ki,t}$, and the processing delay $d^{\text{proc}}_{ki,t}$:
\begin{equation}
    d_{ki,t} =  d^{\text{UL}}_{ki,t} + d^{\text{DL}}_{ki,t} + d^{\text{proc}}_{ki,t}.
\end{equation}
%
\emph{1) Uplink Delay: }
This is the transmission delay of the offloaded task, and is inherently stochastic owing to the randomness in the communication channel and task size \cite{shi2023}:
\begin{equation}
    d^{\text{UL}}_{ki,t} = \frac{s_{k,t}}{B_{ki,t} r_{ki,t}} = \frac{s_{k,t}}{(B/m f_{i,t}) r_{ki,t}},
\end{equation}
where $r_{ki,t}$ is the spectral efficiency of the link between device~ $k$ and edge server~$i$, given by 
\begin{equation}
    r_{ki,t} = \log_2 \left( 1 + \frac{p_0 |g_{ki,t}|^2}{I_{ki,t} + N_0 }\right),
\label{eq:spectral_efficiency}
\end{equation}
where $g_{ki}$ is the parameter for the Rayleigh fading channel, implying that $h_{ki}{=}|g_{ki}|^2$ adheres to an exponential distribution. $p_0$ is the power transmitted by a given device, $N_0$ is the variance of the additive white Gaussian noise, and $I_{ki}$ is the interference power following a uniform distribution \cite{9136780}.\footnote{An extension to mobile users and other fading types is straightforward if the signal-to-noise ratios are independent and identically distributed (i.i.d.) across devices. }

\emph{2) Downlink Delay:} 
It is the transmission time of the result of the offloaded task back to the user. Since the size of the result is very small compared to the size of the offloaded task, we assume that the results are transmitted over a fixed bandwidth $B/ \nu$. The delay is given by
\begin{equation}
    d^{\text{DL}}_{ki,t} {=} \frac{\rho s_{k,t}}{ (B/ \nu) r_{ki,t}}.
    \label{eq:downlink_delay_system}
\end{equation}

\emph{3) Processing Delay:}
Let $F_{ki,t} {=} F/(m f_{i,t})$ be the CPU cycles that edge server~$i$ allocates for the computational task offloaded by device~$k$, and $c$ be the number of processing cycles required for each bit of the offloaded task, then the task processing delay is \cite{shi2023}
\begin{equation}
   d^{\text{proc}}_{ki,t} = \frac{c s_{k,t}}{F_{ki,t}} = \frac{c s_{k,t}}{ F/(m f_{i,t})} .
\label{eq:d_proc_system}
\end{equation}

The task of device~$k$ is successfully offloaded to server~~$i$ at round~$t$ if the total delay does not exceed $d_{\max}$,
\begin{equation}
    d_{ki,t} = \frac{s_{k,t}}{B/(m f_{i,t}) r_{ki,t}} +  \frac{\rho s_{k,t}}{\frac{B}{\nu} r_{ki,t}} +\frac{c s_{k,t}}{F{/}(m f_{i,t})} \le d_{\max}.
\end{equation}
\section{Problem Formulation}
The delays described in Section \ref{sec:system_model} are random due to the uncertainty in the offloaded task sizes and the channel fading. Additionally, the delays depend on the actions of the other devices due to the shared bandwidth and computation resources. The load distribution resulting from the decentralized nature of the problem might be undesirable. Therefore, this paper aims at forcing a target device distribution on the servers. To formulate our problem, let $R_{ki,t}$ be the binary reward obtained by device~$k$ when offloading to server~$i$, 
\begin{equation}
    R_{ki,t} =  \begin{cases} 
      1 & d_{ki,t} \leq d_{\max} \\
      0 & d_{ki,t} > d_{\max}
   \end{cases}.
\end{equation}
Also, denote the target population profile by $\boldsymbol{f}^{*}$. Each device~$k$ decides about offloading so as to maximize its cumulative reward over its lifetime $T$ and achieve $\boldsymbol{f}^{*}$. Let $A_t$ be the arm selected at round $t$, the optimization problem yields
\begin{align}
    \underset{A_t {\in} \{1,\ldots, n\} }{\text{maximize}} \quad \mathbb{E} \left[\Sum{t{=}1}{T} R_{k A_t,t} \right], \\
    \text{s.t.} \quad
    f_{i} = f_{i}^{*},   \quad \forall i \in \{1,\ldots, n\}  .
\end{align}
\section{The Mean Field Model}
Before modelling the problem as a mean field game, we describe the mean field MAB game with binary rewards \cite{gummadi2013mean}.

The game consists of $m$ \emph{agents}, each solving an MAB problem to select one of $n$ available \emph{arms} with unknown reward distributions. Upon selecting an arm, the agent receives a binary reward that depends on the number of other agents selecting the same arm. Additionally, each agent~$k$ has a \emph{type} $\boldsymbol{\theta}_k {\in} [0,1]^n$ that is sampled from a probability distribution $W(\boldsymbol{\theta})$ during the generation of the agent.  The $i$th element in $\boldsymbol{\theta}_k$ is a parameter affecting the agent's reward from arm~$i$. The types are i.i.d. across the agents. 
Additionally, each agent has a known state $\boldsymbol{z}_{k,t} {\in} \mathbb{Z}^{2n}_{+} $ with $2n$ elements representing the agent's total number of wins and losses for each of the $n$ arms over the agent's lifetime.
The state in our task offloading problem is determined by the success in offloading a task from device~$k$ to server~$i$ in round $t$. If $d_{ki,t} {\le} d_{\max}$ then the task is successfully offloaded (win) and element $(2i-1)$ of $\boldsymbol{z}_{k,t}$ is incremented by one. Otherwise, if  $d_{ki,t} > d_{\max}$, the number of losses in element $2i$ is incremented by one. 

The mean field model here assumes that all agents follow the same stationary \emph{policy} to select an arm. Formally, let $\sigma$ be the agent's policy, $\sigma: \mathbb{Z}^{2n}_{+} \rightarrow \{\boldsymbol{x} {\in} [0,1]^n: 
 \sum_{i{=}1}^n x_i {=}1 \}$,  where $\sigma (\boldsymbol{z}_{k,t}, i)$ denotes the probability that an agent chooses arm~$i$ when its current state is $\boldsymbol{z}_{k,t}$, $\sum_{i{=}1}^n \sigma (\boldsymbol{z}_{k,t}, i) {=}1$.  The policy used here is the upper confidence bound  (UCB) \cite{gummadi2013mean}.

To model the process of departure and arrival of agents, they are assumed to have a geometric lifetime, meaning that after each round they regenerate independently with probability $1{-}\beta$.

A key parameter in this model is the population profile, $\boldsymbol{f}_t {=} [f_{1,t}, f_{2,t}, \dots, f_{n,t} ]$, representing the fraction of agents playing each of the $n$ arms at time $t$.
The probability that an agent~$k$ receives a unit reward when selecting an arm~$i$, denoted by $Q({\theta}_{ki}, f_i)$, is a function of the agent's type ${\theta}_{ki}$ and the fraction of agents selecting arm~$i$, $f_i$. These rewards are independent across time, agents, and arms. Here, it is assumed that for a given ${\theta}_{ki}$, $Q({\theta}_{ki}, \cdot)$ is continuous in $f_i$.

Each time step $t$, every agent aims at solving its regret-minimization problem based on $\boldsymbol{z}_{k,t}$. Despite the interaction between agents, the system reaches the mean field steady state (MFSS) with stationary population profile $\boldsymbol{f}$ and state measure $\mu: \mathbb{Z}_{+}^{2n} \to \mathbb{R^{+}}$. The sufficient conditions for the existence of a unique MFSS are stated in \cite[Theorem 1] {gummadi2013mean}. In short: 
\begin{theorem}\label{theorem:mfss-exist}
    Let $L$ be the Lipschitz continuity constant.
    If the following conditions hold for all $a {\in} [0,1]$ and $x, x' {\in} [0,1]$:
    \begin{equation}
        |Q(a,x) -Q(a,x^{'})| \le L |x-x'|,
        \label{eq:condition1}
    \end{equation}
        \begin{equation}
            \beta (1+L) < 1,
        \label{eq:condition2}
    \end{equation}
where $\beta$ is the continuation probability, then for any policy $\sigma$, there exists a unique fixed MFSS.\label{th:theorem1}
\end{theorem}
\section{Task Offloading as A Mean Field MAB Problem}
\label{sec:modelling_task_allocation}
In this section, we model the distributed offloading problem as a mean field MAB game where the offloading devices are the agents, and the edge servers are the arms. 
We present our proposed load-balancing method in Section \ref{sec:load_balancing}. 

In this game, each agent~$k$ decides about offloading without any side information or inter-device communication. A round $t$ corresponds to one of the following two cases:
        
\textbf{Case 1:} A regeneration round with probability $1-\beta$. The agent's state $\boldsymbol{z}_{k,t}$ is reset to zero and its type $\boldsymbol{\theta}_{k,t}$ is sampled from the distribution $W$.
The $i$th element of the type is the normalized signal-to-interference-plus-noise-ratio (SINR):
\begin{equation}
    \theta_{ki,t} = \frac{1}{\gamma_{\max}} \frac{p_0 h_{ki,t}}{I_{ki,t} {+} N_0 },
\end{equation}
where $\gamma_{\max}$ is the maximum SINR to ensure that $\theta_{ki,t} {\in} [0,1]$.
        
\textbf{Case 2:} A continuation round with probability $\beta$. The agent's type does not change, and the agent runs the UCB policy that maps the current state vector to an action (arm selection). As a result, it observes a binary reward that depends on the arm's type and the fraction of users selecting arm~$i$. To show this dependency, we rewrite the spectral efficiency in \eqref{eq:spectral_efficiency} as 
$r_{ki,t}{=}r(\theta_{ki,t}) {=} \log_2 \left(1 {+}\gamma_{\max} \theta_{ki,t} \right)$. The uplink rate of agent~$k$ to server $i$ yields $r^{\text{UL}} (\theta_{ki,t}, f_{i,t}){=} \frac{B}{m f_{i,t}} r(\theta_{ki,t})$. The downlink data rate is given by $r^{\text{DL}} (\theta_{ki}) {=}(B/\nu) r(\theta_{ki})$.
The total delay at time $t$ is calculated as follows, where we remove the subscript $t$ for simplicity:
\begin{equation}
    d(\theta_{ki}, f_i) {=}  \frac{c m s_k f_i}{F} {+} \frac{s_k}{r^{\text{UL}}(\theta_{ki}, f_i)}  {+} \frac{\rho s_k}{r^{\text{DL}}(\theta_{ki})}.
\end{equation}
The binary reward for agent~$k$ for offloading a task to server~$i$ is 1 if $d(\theta_{ki}, f_i) {\le} d_{\max}$ and zero otherwise. The cumulative density function (CDF) of the data size distribution determines the success probability $Q(\theta_{ki}, f_i)$. Given that the packet sizes follow a truncated normal distribution with density proportional to $\mathcal{N(\mu_{\text{s}}, \sigma_{\text{s}})}$  and $s_k {\in} [s_a, s_b], s_a {>} 0$, the CDF yields
\begin{align} 
    Q(\theta_{ki}, f_i) &= \mathbb{P}\left[d(\theta_{ki}, f_i) \le d_{\max}  \right] \nonumber \\
   & = \mathbb{P}\left( s_k \le \frac{d_{\max}F B r(\theta_{ki})}{ m f_i \left(c B r(\theta_{ki}) + F \right) + \rho \nu F}   \right)      \label{eq:q-def} \nonumber \\
   &= \frac{ \Phi' \left(
    \frac{d_{\max}F B r(\theta_{ki}) }{ m f_i \left(c B r(\theta_{ki})   + F\right) + \rho \nu F} 
     \right) -\Phi'( s_a )}
   {\Phi'\left(s_b\right) - \Phi' \left( s_a  \right)},
\end{align}
where $\Phi'(.)$ is the CDF of the standard normal distribution with its argument normalized by $\mu_s$ and $\sigma_s$, $\Phi'(x) {=} \frac{1}{2} \left( 1{+} \text{erf} \left (\frac{x - \mu_s}{\sqrt{2}\sigma_s} \right)\right) $.
\begin{proposition}
    The task allocation mean field MAB game has a unique MFSS.
\end{proposition}
\begin{proof}
    Theorem~\ref{th:theorem1} states that a unique fixed MFSS exists if the conditions in \eqref{eq:condition1} and \eqref{eq:condition2} hold. 
    For a type $\theta_{ki} {\in} [0,1]$ and fractions $f_i, f_i' {\in} [0,1]$ of agents, \eqref{eq:condition1} is expressed as 
$
    \left\lvert   Q(\theta_{ki}, f_i) - Q(\theta_{ki}, f_i')  \right\rvert {\le} L \left\lvert   f_i - f_i'  \right\rvert.
$

As $Q(\theta_{ki}, f_i)$ is differentiable on [0,1], we prove that it is Lipschitz continuous by showing that  $\left\lvert\frac{\partial Q(\theta_{ki}, f_i)}{\partial f_i} \right\rvert {\le} L$ as follows:
\begin{align}\label{eq:ls-cont}
    \left\lvert\frac{\partial Q(\theta_{ki}, f_i)}{\partial f_i} \right\rvert &{=} 
    \frac{1}{\sqrt{2 \pi} \sigma_{\text{s}} Z}
    \frac{m d_{\max} F B r(\theta_{ki}) \left(c B r(\theta_{ki}) {+}F\right)}{\left(m f_i \left(c B r(\theta_{ki}) {+}F\right) {+}\rho \nu F \right)^2}   
    \nonumber \\  & \times
    e^{ - \frac{1}{2 \sigma_{\text{s}}^2}  \left(  \frac{d_{\max}F B r(\theta_{ki})}{m f_i \left(c B r(\theta_{ki}) {+}F\right) {+} \rho \nu F} - \mu_{\text{s}}  \right)^2},
\end{align}
where $Z {=} \Phi'\left(s_b \right) - \Phi'\left(s_a \right)$. The exponential term is bounded, $\left\lvert e^{ - \frac{1}{2 \sigma_{\text{s}}^2}  \left(  \frac{d_{\max}F B r(\theta_{ki})}{m f_i \left(c B r(\theta_{ki}) {+}F\right) {+} \rho \nu F} - \mu_{\text{s}}  \right)^2}  \right\rvert {\le} 1$, and the first term is maximized when $f_i {=} 0$. Therefore, the Lipschitz constant satisfies $L{=}   \frac{1}{\sqrt{2 \pi} \sigma_{\text{s}} Z} \frac{m d_{\max} F B r(\theta_{ki}) \left(c B r(\theta_{ki}) {+}F\right)}{\left(\rho \nu F \right)^2}$. Now, by condition \eqref{eq:condition2} the MFSS is unique when $\beta (1{+}L) < 1.$
\end{proof}
\section{Load-Balancing}
\label{sec:load_balancing}
In this section, we develop a simple load-balancing algorithm that enables the network manager to change the agents' distribution over the servers in a decentralized manner. 

In Section \ref{sec:modelling_task_allocation}, we showed that the offloading problem can be interpreted as a stationary mean field MAB game and proved that a unique MFSS exists for this game. That means, if all parameters remain fixed, the population profile $\boldsymbol{f}_t$ converges to a fixed vector $\boldsymbol{f}$ after a sufficiently large number of rounds $t$. However, the profile $\boldsymbol{f}$ might be undesirable due to practical limitations, such as the servers' cost or accessibility. Thus, we now consider a general form of load-balancing, where we change the arms' rewards such that the resulting population profile of the game closely matches a desired profile $\boldsymbol{f}^{\ast}$. We do this by multiplying the reward probabilities $Q(\theta_{ki},f_i)$ by arm-specific values $\alpha_{i} {\in} [0, 1]$. 
This can be carried out in the mean field setting by changing the type distribution $W$ of the agents, and the resulting system then continues to fulfill the assumptions of the mean field model in Theorem~\ref{th:theorem1}. We formalize this in the following proposition:
\begin{proposition}
\label{prop:type-reward-equiv}
For each $\alpha {\in} [0,1]$ and $\theta {\in} [0,1]$, there exists a corresponding $\theta' {\in} [0,1]$, such that the reward probability is the same,
$\forall f {\in} [0,1]: \; \alpha Q(\theta, f) {=} Q(\theta', f).$
\end{proposition}
\begin{proof}
First, observe from \eqref{eq:q-def} that $Q(0, f) {=} 0$.  
As $\alpha < 1$:
\begin{equation}
    0 = Q(0, f) \leq \alpha Q(\theta, f)  \leq Q(\theta, f).
\end{equation}
Additionally, note from \eqref{eq:ls-cont} that $Q$ is continuous in $\theta$. By the intermediate value theorem, there exists $\theta' {\in} [0, \theta]$ such that 
\begin{equation}
\label{eq:q-equivalence}
    Q(\theta', f) = \alpha Q(\theta, f).
\end{equation}
\end{proof}
Therefore, we can define a function $g_{\boldsymbol{\alpha}}: [0,1]^n {\to} [0,1]^n$ that maps $\theta$ to $\theta'$ and satisfies \eqref{eq:q-equivalence}. The modified type distribution can then be defined as the pushforward $W_{\boldsymbol{\alpha}} {=} g_{\boldsymbol{\alpha} \#} W$.

The population profile $\boldsymbol{f}$ now depends on $\boldsymbol{\alpha}$ and we can formulate the minimization problem as follows:
\begin{equation}
\label{eq:load-balancing-problem}
\min_{\boldsymbol{\alpha} {\in} \Delta^{n}} \frac{1}{2}  \lVert \boldsymbol{f}(\boldsymbol{\alpha}) - \boldsymbol{f^{\ast}} \rVert_2^2,
\end{equation}
where $\Delta^n {=} \{ \alpha {\in} \mathbb{R}_{+}^n {\mid} \sum_i\alpha_i {=} 1 \}$ is the $n$-dimensional unit simplex, a constraint introduced to guarantee a unique solution, as we conjecture $\boldsymbol{f}(\boldsymbol{\alpha})$ to be invariant under multiplication of $\boldsymbol{\alpha}$ by a scaling factor greater than zero. To calculate the stationary profile $\boldsymbol{f}(\boldsymbol{\alpha})$, we resort to a numerical approach. However, to increase efficiency, we make the following observation: 
\begin{proposition}
\label{prop:f-monot}
Consider a multi-agent bandit system $S {=} (m,n,Q,W,\sigma)$ that has a MFSS $(\boldsymbol{f}, \mu)$. Assume that the reward $Q(\theta,f)$ is strictly monotonically decreasing in $f$, and all agents prefer winning arms, i.e., $\sigma$ satisfies
\begin{equation}
\label{eq:efficient-agents}
\sigma(\boldsymbol{z }+ \boldsymbol{w}_j) \geq \sigma(\boldsymbol{z }+ \boldsymbol{l}_j),
\end{equation}
for all $j {\in} \{1,\ldots,m\}$, where $\boldsymbol{w_j} {=} \boldsymbol{e}_{2j{-}1},\; \boldsymbol{l}_j {=} \boldsymbol{e}_{2j}$ are unit vectors corresponding to a win and a loss on arm $j$, respectively. Fix $i {\in} \{1,\ldots,m\}$, then let $\hat S {=} (m,n,Q,\hat W, \sigma)$ with $\hat{W}{=}g_{\boldsymbol{\alpha} \#}W$ and $\alpha_i {\in} \left[ 0,1\right]$, $\alpha_{j \neq i} {=} 1$. Then, $\hat S$ has a MFSS $(\boldsymbol{\hat f}, \hat \mu)$ with $\hat{f}_i {\leq} f_i$.
\begin{proof}
We use Proposition \autoref{prop:type-reward-equiv} to rewrite $\hat{S}$ using the adjusted reward function $\alpha Q$ instead of the adjusted type distribution. For simplicity, we drop the subscript $m$ and the type $\theta$ as both systems now have equal $W$.

The existence of the MFSS $(\boldsymbol{\hat f}, \hat{\mu})$ follows from Theorem~\ref{theorem:mfss-exist}. We continue the proof by contradiction: Assume $\hat f_i > f_i$. Therefore, more agents must pick arm~$i$:
\begin{equation}
    \hat f_i = \mathbb{E}_{z\sim\hat \mu}\left[\sigma(z,i)\right] >
    \mathbb{E}_{z\sim \mu}\left[\sigma(z,i)\right] = f_i.
\end{equation}
As $\sigma$ fulfills (\ref{eq:efficient-agents}), it holds that
$\hat \mu(Z_{x,i}) > \mu(Z_{x,i})  \; \forall x {\in} [0,1]$
where $Z_{x,i} \coloneqq \{ z : \frac{z_{2i{-}1}}{z_{2i{-}1} {+} z_{2i}} {\leq} x\}.$
This requires more agents winning on arm~$i$. As $\alpha Q < Q$ and $Q$ is strictly monotonically decreasing in $f$, $\hat{f_i} < f_i$ must hold to increase the win rate on arm~$i$, which contradicts $\hat f_i > f$.
\end{proof}
\end{proposition}
By using UCB as $\sigma$ and $Q$ as in \eqref{eq:q-def}, the conditions in Proposition \ref{prop:f-monot} are fulfilled. Therefore, if we obtain a population profile $\boldsymbol{f}(\boldsymbol{\alpha})$ using a numerical method, we at least know in which direction to move $\boldsymbol{\alpha}$ to achieve a value closer to the desired profile $\boldsymbol{f^\ast}$.

Based on the discussion above, we propose to optimize \eqref{eq:load-balancing-problem} using a variant of projected gradient descent with a fixed step size $s$ and a rough estimate of the gradient of $\alpha$:
\begin{equation}
\nabla \boldsymbol{\alpha}_t {\approx}  \boldsymbol{f}(\boldsymbol{\alpha}_t) {-} \boldsymbol{f^\ast}, \quad
\boldsymbol{\alpha}_{t{+}1} {=} P_{\Delta^n}(\boldsymbol{\alpha}_t {-} s\cdot \nabla \boldsymbol{\alpha}_t),
\end{equation}
where $P_{\Delta^n}$ is the projection on the unit simplex. 

Algorithm \ref{alg:pgd} summarizes our proposed method. In Section \ref{sec:numerical-results}, numerical results show that our algorithm quickly converges to a near-optimal solution.
\begin{algorithm}
    \SetKwInOut{Input}{Input}
    \SetKwInOut{Output}{Output}
    \caption{Load-balancing algorithm}
    \label{alg:pgd}

    \textbf{def} loadBalancing($\boldsymbol{f}^\ast$, $s$):\\
    \Input{Desired population profile $\boldsymbol{f}^\ast$ and step size $s$}
    \Output{Set of parameters $\boldsymbol{\alpha}$ such that $\boldsymbol{f}(\boldsymbol{\alpha}) \approx \boldsymbol{f}^{\ast}$}
    $\boldsymbol{\alpha} \gets [1/n,\ldots,1/n]$ \\
    \While{not converged}{
        $\boldsymbol{f}_{t} \gets \text{simulateMAB}(\boldsymbol{\alpha}) \label{alg:mab}$\\
        $\nabla \boldsymbol{\alpha} \gets \boldsymbol{f}_t {-} \boldsymbol{f^\ast}$ \\
        $\boldsymbol{\alpha} \gets \text{projectSimplex}(\boldsymbol{\alpha} {-} s\cdot \nabla \boldsymbol{\alpha})$
    }
    return $\boldsymbol{\alpha}$
\end{algorithm}
\section{Numerical Results}
\label{sec:numerical-results}
In this section, we conduct four experiments to illustrate the convergence of the population profile to the MFSS and show the effectiveness of the proposed load-balancing algorithm. We consider a network with $m{\in}\{100,10^4\}$ agents and $n{\in}\{2,8\}$ arms. Each agent transmits tasks of random size $s_k {\in} [0.5, 1]  $ Mbit over a Rayleigh fading channel with parameter $1$  \cite{chen2019mobile}. The transmission power is $p_0 {=}200$ mW  \cite{chen2019mobile}, the noise power density is ${-}174$ dBm, and $I_{ki} {\in} [8,12]$ mW \cite{9136780}. The 
bandwidth $B{=}10$ MHz \cite{chen2019mobile} is shared between devices transmitting to the same server. Similarly, the available CPU cycles, $F{=}4$ GHz \cite{chen2019mobile}, are divided among the offloading devices. For $n {=} 2$, we choose $\boldsymbol{f}^{\ast} {=} \begin{bmatrix} 0.2 , 0.8 \end{bmatrix}$, for $n{=}8$ we choose $\boldsymbol{f}^{\ast} {=} \begin{bmatrix}  0.07 , 0.08 , \ldots , 0.13 , 0.3 \end{bmatrix}$. One can select the profile flexibly as long as the values are not extreme. 

We select the maximum delay as $d_{\text{max}} {=}2m \cdot(0.26)$. This value is the average delay per agent, which we obtained empirically when simulating the communication model, multiplied by $2m$ to accelerate the convergence of our algorithm. Since we require $\boldsymbol{\alpha}$ to sum to $1$, the reward of each server reduces. If $d_{\text{max}}$ is too low, the server gives very few rewards in total, and tuning $\alpha$ for this server (further reducing the rewards) does not change the agent choice too much.

To obtain the desired $\boldsymbol{\alpha}$ for the target profile, Algorithm \ref{alg:pgd} is executed for $150$ optimization steps. In \autoref{alg:mab}, the MAB game is simulated for $300$ steps. To stabilize the optimization, we choose $\boldsymbol{f}_t$ as the average profile over the last $50$ steps. This enables obtaining good convergence even for only $100$ agents, where the profile is still very noisy in equilibrium.

Fig.\ref{fig:1} shows the performance of our optimization procedure for a network with $2$ arms, tested with 100 agents in one scenario and $10^4$ agents in another. It depicts the fraction of agents choosing arm $1$ over $300$ simulation steps, comparing the unadjusted setting (blue) against the adjusted setting (orange) obtained after running our optimization algorithm. In the first case, because all fading channels are i.i.d., the devices are evenly distributed across the servers. After adjusting the population profile, the fraction of devices choosing the first server (arm) converges to the target value of $f^*_1{=}0.2$.

Fig. \ref{fig:2} provides the training curves of Algorithm~\ref{alg:pgd}. It reports $\lVert \boldsymbol{f}^{\ast} {-} \boldsymbol{f}_t \rVert_2^2$, where $\boldsymbol{f}_t$ is the profile obtained by averaging the last $50$ steps of a $300$-step MAB simulation. We observe good convergence in all four experiments.
\begin{figure}[h]
\centering
\includegraphics[width=0.48\textwidth]{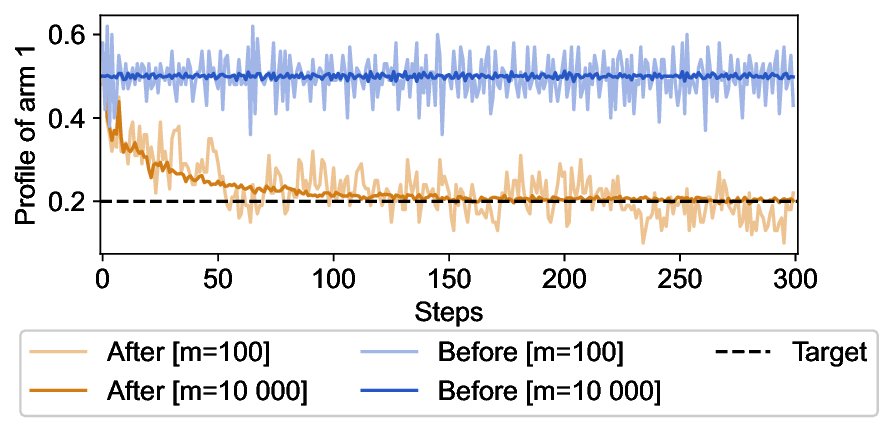}
\caption{The fraction of agents selecting arm 1, $f_1$, for $100$ and $10^4$ agents before (blue) and after (orange) applying the load-balancing algorithm. The value of $f_1^{*}$ is $0.2$, and there are $300$ optimziation steps.}
\label{fig:1}
\end{figure}
\begin{figure}[h]
\centering
\includegraphics[width=0.44\textwidth]{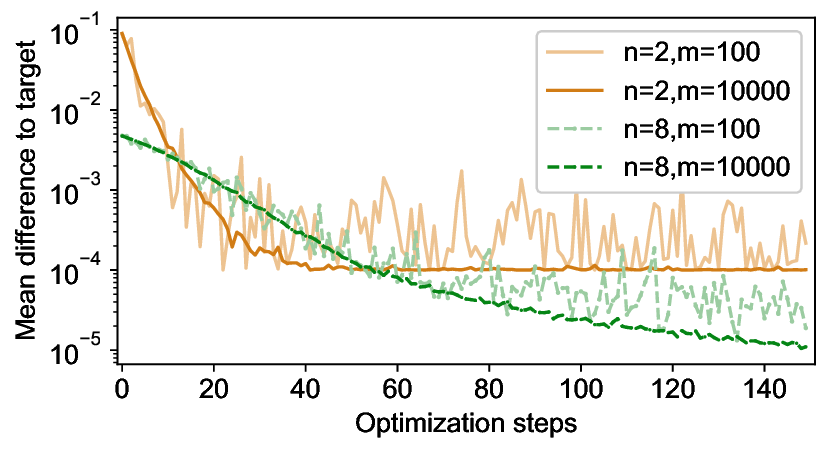}
\caption{ The mean distance between $f$ and $f^{*}$ over 150 optimization steps for networks with two (orange) and eight (green) servers with $100$ and $10^4$ agents in each case.}
\label{fig:2}
\end{figure}
\section{Conclusion}
We investigated a decentralized competitive task offloading and load balancing problem under uncertainty about available communication- and computation resources and tasks' size. We modeled it as a mean field MAB game and proved it has a unique MFSS. We proposed a novel decentralized decision-making policy that guarantees the convergence of servers' load to a target population profile according to theoretical- and numerical analysis.

Future work could extend this research in several directions. One of them is to investigate a more complicated system model, for example, the impact of interference on the system performance and its convergence to the MFSS. Furthermore, one can relax the assumption of a known target population profile and explore designing $f^{*}$ to minimize offloading costs or ensure fairness across servers. Lastly, we plan to conduct a more in-depth theoretical analysis of the proposed mean field problem. Of particular interest are theoretical guarantees for the convergence of our proposed algorithm, for example, by analyzing the convexity of $f(\alpha)$ and using some methods to directly estimate $f(\alpha)$ from $\alpha$, which might accelerate the simulation process.
\bibliography{bibliography_MEC}
\bibliographystyle{ieeetr}

\end{document}